\documentclass[a4paper,12pt,notitlepage]{article}
\usepackage[utf8]{inputenc}
\usepackage[T1]{fontenc}
\usepackage{fullpage}
\usepackage{times}
\usepackage{amsmath}
\usepackage{amsfonts}
\usepackage{amssymb}
\usepackage[makeroom]{cancel}
\usepackage{amsthm}
\usepackage{graphicx}
\usepackage{float}
\usepackage{multirow}
\usepackage{mathtools}
\usepackage{bbm}
\usepackage{braket}
\usepackage{ytableau,tikz,varwidth}
\usepackage[enableskew]{youngtab}
\usepackage{hyphenat}
\usepackage{verbatim}
\usepackage{subcaption}
\usepackage{hyperref}
\usepackage{soul}
\usepackage{etoolbox}
\usepackage{nameref}
\usetikzlibrary{tikzmark,decorations.pathreplacing,calc}

\usepackage[backend = bibtex, sorting=none]{biblatex}
\theoremstyle{definition}

\allowdisplaybreaks

\DeclarePairedDelimiterX\bk[2]{\langle}{\rangle}{#1 \delimsize\vert #2}
\DeclarePairedDelimiterX\kb[2]{\vert }{\vert }{#1 \delimsize\rangle\langle#2}
\DeclarePairedDelimiter\ceil{\lceil}{\rceil}

\def\={\;=\;} \def\+{\,+\,}

\newcommand\restr[2]{{
  \left.\kern-\nulldelimiterspace 
  #1 
  \vphantom{\big|} 
  \right|_{#2} 
  }}

\newtheorem{theorem}{Theorem}
\newtheorem{lemma}[theorem]{Lemma}

\newtheorem{example}{Example}

\usepackage{tikz}

\apptocmd\appendix{%
  \addcontentsline{toc}{chapter}{Appendix}%
  \counterwithin{equation}{section}%
  \counterwithin{figure}{section}%
  \counterwithin{table}{section}%
}{}{}

\makeatletter
\newcommand{\cbigoplus}{\DOTSB\cbigoplus@\slimits@}
\newcommand{\cbigoplus@}{\mathop{\widehat{\bigoplus}}}
\makeatother

\begin{document}

\title{Calculable lower bounds on the efficiency of universal sets of quantum gates}
\author{Oskar S\l{}owik$^{*}$, Adam Sawicki$^{\dagger}$}

\date{%
     Center for Theoretical Physics, Polish Academy of Sciences,\\ Al. Lotników 32/46, 02-668 Warsaw, Poland\\%
$^*$oslowik@cft.edu.pl, $^\dagger$a.sawicki@cft.edu.pl\\[2ex]%
}
\maketitle

\begin{abstract}
Currently available quantum computers, so called Noisy Intermediate-Scale Quantum (NISQ) devices, are characterized by relatively low number of qubits and moderate gate fidelities. In such scenario, the implementation of quantum error correction is impossible and the performance of those devices is quite modest. In particular, the depth of circuits implementable with reasonably high fidelity is limited, and the minimization of circuit depth is required. Such depths depend on the efficiency of the universal set of gates $\mathcal{S}$ used in computation, and can be bounded using the Solovay-Kitaev theorem. However, it is known that much better, asymptotically tight bounds of the form $\mathcal{O}(\mathrm{log}(\epsilon^{-1}))$, can be obtained for specific $\mathcal{S}$. Those bounds are controlled by so called spectral gap, denoted $\mathrm{gap}(\mathcal{S})$. Yet, the computation of $\mathrm{gap}(\mathcal{S})$ is not possible for general $\mathcal{S}$ and in practice one considers spectral gap at a certain scale $r(\epsilon)$, denoted $\mathrm{gap}_r(\mathcal{S})$. This turns out to be sufficient to bound the efficiency of $\mathcal{S}$ provided that one is interested in a physically feasible case, in which an error $\epsilon$ is bounded from below. In this paper we derive lower bounds on $\mathrm{gap}_r(\mathcal{S})$ and, as a consequence, on the efficiency of universal sets of $d$-dimensional quantum gates $\mathcal{S}$ satisfying an additional condition. The condition is naturally met for generic quantum gates, such as e.g. Haar random gates. Our bounds are explicit in the sense that all parameters can be determined by numerical calculations on existing computers, at least for small $d$. This is in contrast with known lower bounds on $\mathrm{gap}_r(\mathcal{S})$ which involve parameters with ambiguous values.
\end{abstract}

\section{Introduction and main results}

Universal, scalable and fault-tolerant quantum computers are the holy grail of quantum computing. Such devices require quantum error correction that, due to quantum threshold theorem, can be implemented if the levels of gate errors are small enough \cite{knill1998, kitaev2003, aharonov2008}. However, recent quantum hardware, so called Noisy Intermediate-Scale Quantum (NISQ) devices, does not offer gate fidelities required for quantum error correction and their performance is heavily affected by gate imperfections \cite{montanaro2017, preskill2018, boxio2018}. Because of error accumulation effects, the depth of circuits feasible for NISQ devices is very modest. Hence it is imperative to find ways to minimize such depths. One of the ways to address this issue is to focus on the efficiency of universal sets \cite{sawicki_s17, sawicki17} of gates $\mathcal{S}$ used for the computations.

Spectral gap is an useful measure of efficiency of universal sets of quantum gates $\mathcal{S} \subset \mathrm{SU}(d)$. The value of gap for chosen $\mathcal{S}$, denoted $\mathrm{gap}(\mathcal{S})$, lies between 0 (no gap) and some optimal value $\mathrm{gap}_{\mathrm{opt}} < 1$, depending only on the number of gates $|\mathcal{S}|$ \cite{kesten59}. The higher the value of $\mathrm{gap}(\mathcal{S})$, the better is the upper bound on the minimal length (circuit depth) $\ell$ of a sequence of gates from $\mathcal{S}$ required to $\epsilon$-approximate any unitary operation from $\mathrm{SU}(d)$.
Recall that the Solovay-Kitaev theorem \cite{kitaev2002} provides such a bound for depth $\ell$,
\begin{equation}
\ell = \tilde{A}(\mathcal{S}) \cdot \mathrm{log}^{3+\delta}(1/\epsilon)), \quad \delta > 0 \mathcal{.}
\end{equation}
However, the existence of gap, i.e. $\mathrm{gap}(\mathcal{S}) > 0$, implies that
\begin{equation}
\label{eq:ah}
\ell = A(\mathcal{S}) \cdot \mathrm{log}(1/\epsilon)) + B(\mathcal{S})    
\end{equation}
is enough, with the constants $A$ and $B$ proportional to $\mathrm{log}^{-1}(1/(1-\mathrm{gap}(\mathcal{S})))$ \cite{harrow02}. In fact, $\ell= \mathcal{O}(\mathrm{log}(1/\epsilon)))$ is optimal, which can be seen from a simple volumetric argument. 


One should note that some properties of $\mathcal{S}$ with optimal spectral gap are known. For instance, if the gates from the universal set $\mathcal{S}$ have algebraic entries then the gap exists \cite{gamburd08, gamburd12}. Moreover, it has been conjectured that any universal $\mathcal{S}$ has the gap and there are explicit constructions of examples of $\mathcal{S}$ with the optimal spectral gap for $\mathrm{SU}(2)$ with $|\mathcal{S}|= p -1$ for $p \equiv 1 \, \mathrm{mod} \, 4$ \cite{lps86, lps87}. Finally, some commonly used one-qubit universal sets turned out to have the optimal spectral gap \cite{bocharov2013, selinger2015, sarnak2015, kliuch2016}. However, the construction of many-qubit gates with the optimal spectral gap remains an open problem.

The calculation of $\mathrm{gap}(\mathcal{S})$ is challenging and in practice one often considers the gap up to the certain scale $r$, denoted $\mathrm{gap}_r(\mathcal{S})$, such that $\mathrm{gap}(\mathcal{S})$ is the infimum of $\mathrm{gap}_r(\mathcal{S})$ over all scales $r$ \footnote{The exact definition of a scale depends on the approach.}. Since it is impossible to implement gates without any error, in practice $\epsilon$ can be  bounded from below. In such a case, in order to bound $\ell$ it is sufficient to have the knowledge of $\mathrm{gap}_r(\mathcal{S})$ at some scale $r(\epsilon)$ instead of $\mathrm{gap}(\mathcal{S})$. This is due to the existence of the Solovay-Kitaev-like theorems involving $\mathrm{gap}_r(\mathcal{S})$. Specifically, it is known that for any universal $\mathcal{S}$ one can bound $\ell \propto \mathrm{gap}_r^{-1}(\mathcal{S}) \cdot  \mathrm{log}(1/\epsilon)$ at some scale $r(\epsilon)$ (see the first part of Lemma 5 in \cite{varju13} and the improved version with $r(\epsilon) \simeq \mathcal{O}(1/\epsilon \cdot \mathrm{log}(1/\epsilon))$ - Proposition 2 in \cite{oszmaniec22}). Thus, bounding $\mathrm{gap}_r(\mathcal{S})$ is imperative. From the seminal paper \cite{varju13} it is known, in more general setting of semisimple compact connected Lie groups, that there exist group constants $c, A$ and $r_0$ such that
\begin{equation}
\label{eq:vb}
\mathrm{gap}_r(\mathcal{S}) \geq c  \cdot \mathrm{gap}_{r_0}(\mathcal{S}) \cdot  \mathrm{log}^{-A} (r) \mathrm{,}   
\end{equation}
for any $r \geq r_0$. Thus, the knowledge of gap at the certain scale $r_0$ enables to bound the rate at which $\mathrm{gap}_r(\mathcal{S})$ vanishes with growing $r \geq r_0$.  However it is unclear what is the magnitude of the minimal scale $r_0$ from which the bound (\ref{eq:vb}) holds, even for $\mathrm{SU}(2)$. Our preliminary analysis of this bound suggests that the value of $r_0$ for $\mathrm{SU}(2)$ resulting from the proof is enormous - orders of magnitude larger than the scale for which the numerical calculation of $\mathrm{gap}_{r_0}(\mathcal{S})$ is remotely possible.


In this paper we exploit Bourgain's argument for bounding $\mathrm{gap}_{r}(\mathcal{S})$ by the diameter of $\mathcal{S}$, which was communicated in the proof of the second part of Lemma 5 from \cite{varju13}. By introducing an additional assumption on $\mathcal{S}$ we obtain calculable bounds on $\mathrm{gap}_{r}(\mathcal{S})$ for universal sets of quantum gates. Our additional assumption on $\mathcal{S}$ is satisfied e.g. for generic quantum gates, such as Haar random gates (with probability 1). The main result of the paper is the following.

\begin{theorem}
\label{th:main}
Let $\mathcal{S}=\{U_1, \ldots, U_k, U_1^{-1}, \ldots U_k^{-1}\}$ be a universal symmetric set of $d$-dimensional quantum gates, such that for any $U_i, U_j \in \mathcal{S}$, $i \neq j$, the set $\{U_i^2, U_j^2, U_i^{-2}, U_j^{-2}\}$ is universal. Then
\begin{equation}
\label{eq:main}
\mathrm{gap}_t(\mathcal{S}) \geq \alpha \cdot g_{t_0}(\mathcal{S}) \cdot \mathrm{log}^{-2c}(\beta t)\mathrm{,}
\end{equation}
where $c=\mathrm{log}(5)/\mathrm{log}(3/2) \approx 4$, $\alpha$ and $\beta$ are known constants and $g_{t_0}(\mathcal{S})$ can be determined by the numerical calculations of gaps at a known scale $t_0$ of certain universal sets that can be derived from $\mathcal{S}$.
\end{theorem}

The quantity $g_{t_0}(\mathcal{S})$ is defined in equation (\ref{eq:gt0}), see also (\ref{eq:gapt0r}) and (\ref{eq:sind}). Crucially, we provide explicit formulas (\ref{eq:ab}) and (\ref{eq:t0}), (\ref{eq:tau}) for $\alpha=\alpha(d, \epsilon_0)$, $\beta=\beta(d)$ and $t_0=t_0(d, \epsilon_0)$, where $\epsilon_0$ is the parameter in the construction leading to different bounds. The value of $t_0$ is small enough to enable numerical calculations of $g_{t_0}(\mathcal{S})$, at least for $d=2,3$ and $4$. Hence, our bounds can be made explicit by numerical experiments for fixed $\mathcal{S}$. We provide examples of specific values of $t_0$, $\alpha$ and $\beta$, for $d=2, 3$ and $4$ in Tables \ref{tab:d23} and \ref{tab:d4}. The minimal possible values of $t_0$ are indicated by bold font and given by
\begin{equation}
t_{\mathrm{min}} \coloneqq \ceil{5 d^{5/2} /\epsilon_{0, \mathrm{max}} \cdot \tau(\epsilon_{0, \mathrm{max}},d)} \mathrm{,} \quad \epsilon_{0, \mathrm{max}} \coloneqq 1/(d+2) \mathrm{,}
\end{equation}
where $\tau(\epsilon, d)$ is defined in (\ref{eq:tau}). We present the values of $t_0$ up to the ones giving $\alpha$ around 1.

\begin{table}[H]
\begin{center}
\begin{tabular}{ |c|c|c|c| } 
 \hline
  $\epsilon_0$ & $t_0$ & $\alpha$ & $\beta$ \\ \hline 
 0.04 & 4599 & 9.68e-01 & 0.393\\ \hline
 0.05 & 3544 & 3.67e-01 & 0.393\\ \hline
 0.06 & 2860 & 1.49e-01 & 0.393\\ \hline
 0.07 & 2384 & 6.29e-02 & 0.393\\ \hline
 0.08 & 2035 & 2.72e-02 & 0.393\\ \hline
0.09 & 1769 & 1.18e-02 & 0.393\\ \hline
 0.1 & 1559 & 5.15e-03 & 0.393\\ \hline
 0.11 & 1391 & 2.22e-03 & 0.393\\ \hline
 0.12 & 1252 & 9.38e-04 & 0.393\\ \hline
 0.13 & 1137 & 3.85e-04 & 0.393\\ \hline
 0.14 & 1039 & 2.62e-04 & 0.393\\ \hline
 0.15 & 955 & 5.68e-05 & 0.393\\ \hline
 0.16 & 883 & 1.99e-05 & 0.393\\ \hline
 0.17 & 820 & 6.36e-06 & 0.393\\ \hline
 $ \lessapprox$ \textbf{0.25} & \textbf{509} &   $\gtrapprox$ \textbf{0} & \textbf{0.393} \\ \hline
\end{tabular}, \quad \begin{tabular}{ |c|c|c|c| } 
 \hline
  $\epsilon_0$ & $t_0$ & $\alpha$ & $\beta$ \\ \hline 
0.02 & 29199 & 7.25e-01 & 0.251\\ \hline
 0.03 & 18353 & 1.73e-01 & 0.251\\ \hline
 0.04 & 13170 & 5.07e-02 & 0.251\\ \hline
0.05 & 10166 & 1.65e-02 & 0.251\\ \hline
 0.06 & 8219 & 5.69e-03 & 0.251\\ \hline
 0.07 & 6861 & 2.01e-03 & 0.251\\ \hline
 0.08 & 5864 & 7.10e-04 & 0.251\\ \hline
0.09 & 5103 & 2.47e-04& 0.251\\ \hline
 0.1 & 4504 & 8.31e-05 & 0.251\\ \hline
 0.11 & 4022 & 2.65e-05 & 0.251\\ \hline
 0.12 & 3625 & 7.81e-06 & 0.251\\ \hline
 0.13 & 3295 & 2.07e-06 & 0.251\\ \hline
 0.14 & 3014 & 4.75e-07 & 0.251\\ \hline
 0.15 & 2775 & 8.82e-08 & 0.251\\ \hline
 $\lessapprox$ \textbf{0.20} & \textbf{1958} &   $\gtrapprox$ \textbf{0} & \textbf{0.251}\\ \hline
\end{tabular}
\end{center}
\caption{Examples of values of $t_0$, $\alpha$ and $\beta$ for $d=2$ (left) and $d=3$ (right). The parameter $\epsilon_0$ is an element of the construction determining $t_0$ (along with $d$). Bold font indicates the choice of the smallest possible $t_0$.}
\label{tab:d23}
\end{table}

\begin{table}[H]
\begin{center}
\begin{tabular}{ |c|c|c|c| } 
 \hline
  $\epsilon_0$ & $t_0$ & $\alpha$ & $\beta$ \\ \hline 
 0.01 & 134232 & 8.46e-01 & 0.175\\ \hline
 0.02 & 61313 & 1.06e-01 & 0.175\\ \hline
 0.03 & 38602 & 2.18e-02 & 0.175\\ \hline
 0.04 & 27738 & 5.49e-03 & 0.175\\ \hline
 0.05 & 21435 & 1.52e-03 & 0.175\\ \hline
 0.06 & 17347 & 4.34e-04 & 0.175\\ \hline
 0.07 & 14494 & 1.24e-04 & 0.175\\ \hline
 0.08 & 12398 & 3.43e-05 & 0.175\\ \hline
 0.09 & 10797 & 8.86e-06 & 0.175\\ \hline
 0.1 & 9537 & 2.07e-06 & 0.175\\ \hline
 0.11 & 8522 & 4.14e-07 & 0.175\\ \hline
 0.12 & 7687 & 6.59e-08 & 0.175\\ \hline
 0.13 & 6990 & 7.37e-09 & 0.175\\ \hline
 0.14 & 6400 & 4.54e-10 & 0.175\\ \hline
  $\lessapprox$ \textbf{0.1(6)} & \textbf{5195} &   $\gtrapprox$ \textbf{0} & \textbf{0.175}\\ \hline
\end{tabular}
\end{center}
\caption{Examples of values of $t_0$, $\alpha$ and $\beta$ for $d=4$. The parameter $\epsilon_0$ is an element of the construction determining $t_0$ (along with $d$). Bold font indicates the choice of the smallest possible $t_0$.}
\label{tab:d4}
\end{table}

The value of $\alpha$ grows quickly with $t_0$ as can be seen in Fig. \ref{fig:alpha}. Values of $\beta$ and $c$ do not depend on $t_0$. On the other hand, the value of $g_{t_0}(\mathcal{S})$ can decrease with increasing $t_0$.

\begin{center}
\includegraphics[width=0.7\linewidth]{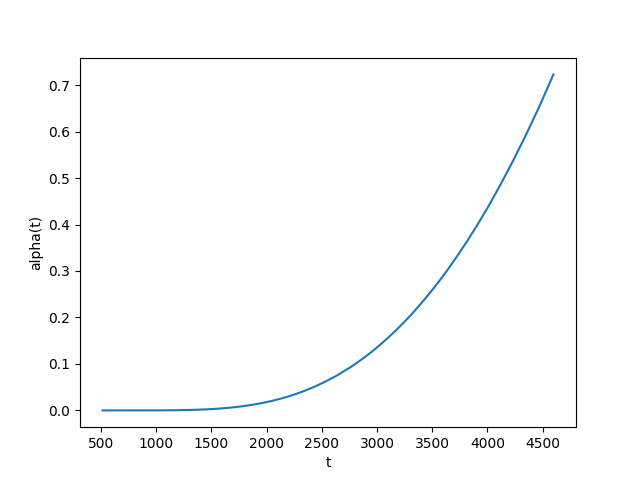}
\captionof{figure}{The value of $\alpha$ as a function of $t_0$ for $d=2$.}
\label{fig:alpha}
\end{center}

In order to check the behaviour of our bound (\ref{eq:main}) and demonstrate that it can be calculated on existing hardware, we performed a numerical simulation on a supercomputer. For the sake of this simulation, we chose 1000 Haar random sets $\mathcal{S}$ for $d=2$, each consisting of three gates and their inverses. The computations took approximately two weeks and utilized 1008 CPU cores. We calculated the values of the lower bound for $t_0$ ranging from 550 to 900 (with increment 10) and plotted the bounds for $t$ from $t_0$ to 1000. We also calculated the ratio of our bound and the true value of the gap at given $t$. We present those results averaged over all sets $\mathcal{S}$ in Fig. \ref{fig:gap_ratio}.

\begin{center}
\includegraphics[width=0.7\linewidth]{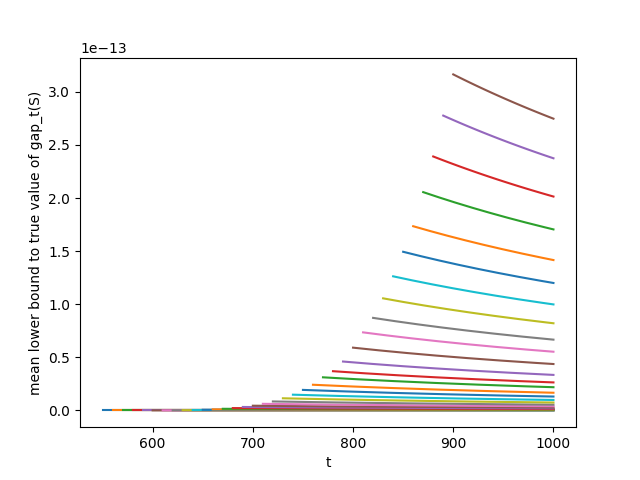}
\captionof{figure}{Ratio of lower bound (\ref{eq:main}) to true value of gap as a function of $t \in [t_0, 1000]$. The ratio was averaged over 1000 Haar random sets $\mathcal{S}$ with $d=2$ and 3 gates on each set together with inverses. Each line corresponds to a lower bound calculated for different value of $t_0 \in [550, 900]$ with increment 10.}
\label{fig:gap_ratio}
\end{center}
The value of the lower bound looks qualitatively the same as the ratio in Fig. \ref{fig:gap_ratio} rescaled by a constant. This is because the true value of the gap is practically constant for any chosen $\mathcal{S}$ in the inspected range of $t$. From Fig. \ref{fig:gap_ratio} it is clear that our bound is far from being tight, at least in a tested range. However, obtained results are not far from our expectations taking into account the generality of our bounds. Moreover, evidently, the lower bounds improve quickly with $t_0$, due to the rising value of constant $\alpha$ that dominates possible deterioration of $g_{t_0}(\mathcal{S})$. In fact, the value of $g_{t_0}(\mathcal{S})$ is also constant for any $\mathcal{S}$ in the inspected range. Needless to say, such improvement cannot continue indefinitely, since the ratio must be at most 1. Unfortunately, we didn't have enough resources to push our simulations further.

The structure of this paper is as follows. In Section \ref{sec:ao} we introduce the mathematics used in the paper, such as the averaging operators and their relevant spectral gaps. In Section \ref{sec:bog} we provide an alternative proof of the efficiency bound (\ref{eq:ah}) from \cite{harrow02} with $A$ and $B$ proportional to $\mathrm{gap}^{-1}(\mathcal{S})$. In Section \ref{sec:mth} we present the proof of our main result, Theorem \ref{th:main}. 

\section{Averaging operators and their spectral gaps}
\label{sec:ao}
By $G_d$ we denote the projective unitary group $\mathrm{PU}(d)$, which is the quotient of the unitary group $\mathrm{U}(d)$ by its center \begin{equation}
 \mathrm{U}(1)=\{e^{i \theta}|\, \theta \in [0, 2 \pi) \} \mathrm{.}  
\end{equation}
Consider the space $L^2(G_d)$ of square integrable complex functions on $G_d$ with respect to $\mu$, equipped with the standard scalar product $\langle \cdot, \cdot \rangle$ (linear on the second slot). Since $G_d$ is compact we consider only unitary representations. A group $G_d$ acts on $L^2(G_d)$ via (left) regular representation $\mathrm{Reg}$. Given a function $f \in L^2(G_d)$ and element $g \in G_d$
\begin{equation}
(\mathrm{Reg}(g) f)(x)=f(g^{-1}x) \mathrm{,}    
\end{equation}
so the regular representation acts on functions by shifts.
Regular representation is not irreducible. In fact, due to Peter-Weyl theorem, it decomposes into an orthogonal direct sum of all the irreducible unitary representations (irreps) with multiplicities equal to their dimensions
\begin{equation}
L^2(G_d) = \cbigoplus_{\lambda \in \Lambda} V_{\lambda}^{ \oplus d_{\lambda} } \mathrm{,}  
\end{equation}
where $\Lambda$ is the set of highest weights of $G_d$ (enumerating all irreps up to isomorphism), $V_{\lambda}$ is the representation space of irrep $\pi_{\lambda}$ with highest weight $\lambda$ and dimension $d_{\lambda}$ and hat denotes the closure of an infinite direct sum. Moreover for each $\lambda \in \Lambda$
\begin{equation}
\left\{ \sqrt{d_{\lambda}}(\pi_{\lambda})_{ij}| 1\leq i, j \leq d_{\lambda}  \right\} \mathrm{,}
\end{equation}
is an orthonormal basis of $V_{\lambda}$ where matrix elements $(\pi_{\lambda})_{ij}$ are functions in $L^2(G_d)$ given by
\begin{equation}
(\pi_{\lambda})_{ij}(g) \coloneqq \langle e_i, \pi_{\lambda}(g) e_j \rangle \mathrm{,}
\end{equation}
for some fixed orthonormal basis of $V_{\lambda}$, $\{e_k | 1 \leq k \leq d_{\lambda}\}$.
Clearly, sum of all such basis form an orthonormal basis of $L^2(G_d)$
\begin{equation}
\label{eq:fbasis}
\left\{ \sqrt{d_{\lambda}}(\pi_{\lambda})_{ij}| \lambda \in \Lambda, 1\leq i, j \leq d_{\lambda}  \right\} \mathrm{.}
\end{equation}
Hence any function $f \in V_{\lambda}$, as a linear combination of matrix elements, is given by 
\begin{equation}
f(g) = \mathrm{Tr}[A \pi_{\lambda}(g)], \quad g \in G_d
\end{equation}
for some complex $d_{\lambda} \times d_{\lambda}$ matrix $A$. The regular representation restricted to functions in $V_{\lambda}$ is isomorphic to representation $\pi_{\lambda}$. If $\pi$ is any (possibly reducible) representation of $G_d$, then $\pi$ is isomorphic to the direct sum of irreps which can be identified with function spaces from Peter-Weyl decomposition $V_{\lambda_1}^{\otimes m_1} \oplus \ldots \oplus V_{\lambda_k}^{\otimes m_k}$, for some $k \geq 1$ and multiplicities $m_i \geq 1$. If $m_i \leq d_{\lambda_i}$ for all $1 \leq i \leq k$, then representation $\pi$ will appear as a subrepresentation of $L^2(G)$. The corresponding space of functions consists of all functions obtained via
\begin{equation}
f(g) = \mathrm{Tr}[A \pi(g)], \quad g \in G_d
\end{equation}
for all matrices $A$.

We now comment on how one may naturally choose a scale up to which one would like to consider irreps of $G_d$.

The Lie algebra of $G_d$ is isomorphic to $\mathfrak{su}(d)$ since \begin{equation}
 G_d = \mathrm{PU}(d)\cong \mathrm{PSU(d)}=\mathrm{SU(d)}/Z(\mathrm{SU}(d)) \mathrm{,}   
\end{equation} where $Z(\mathrm{SU}(d)) \simeq \mathbb{Z}_d$, the center of $\mathrm{SU}(d)$, is discrete.

The adjoint representation $\mathrm{Ad}$ of $\mathrm{U}(d)$ descents into the quotient group $G_d$ forming the adjoint representation $\mathrm{Ad}$ of a group $G_d$ acting on its representation space $\mathfrak{su}(d)$ via
\begin{equation}
\mathrm{Ad}_{U}(X)=\hat{U} X \hat{U}^{-1}, \quad U \in G_d, X \in \mathfrak{su}(d) \mathrm{,}   
\end{equation}
where $\hat{U} \in \mathrm{U}(d)$ is any representative of $U \in G_d$. Importantly $\mathrm{Ad}$ is faithful hence every representation of $G_d$ is realized inside $\mathrm{Ad}^{\otimes n}$ for $n$ large enough. The defining representation $U$ of $\mathrm{U}(d)$ does not descend into a well-defined representation $U$ of $G_d$ but $U \otimes \overline{U}$ does, where $\overline{U}$ is the adjoint of $U$. In fact,
\begin{equation}
U \otimes \overline{U} \cong \mathrm{Ad} \oplus I \mathrm{,}    
\end{equation}
where by $I$ we denote the one-dimensional trivial representation.
Thus each irrep of $G_d$ appears in rep
$(U \otimes \overline{U})^{\otimes t}$
for some $t$. Moreover this rep contains only projective irreps of $\mathrm{U}(d)$ hence reps of $G_d$.

Consider $t \geq 2$. Then
\begin{equation}
(U \otimes \overline{U})^{\otimes t} \cong (U \otimes \overline{U})^{\otimes t-1} \otimes (\mathrm{Ad} \oplus I) \cong [(U \otimes \overline{U})^{\otimes t-1} \otimes \mathrm{Ad}] \oplus (U \otimes \overline{U})^{\otimes t-1} 
\end{equation}
and applying this reasoning inductively we see that all irreps of $(U \otimes \overline{U})^{\otimes s}$ appear in $(U \otimes \overline{U})^{\otimes t}$ for $s \leq t$. Thus we see that with $t$ increasing the rep $(U \otimes \overline{U})^{\otimes t}$ contains more and more irreps of $G_d$ and each irrep of $G_d$ is contained in this rep for $t$ large enough. In the language of Peter-Weyl theorem the corresponding functions in $L^2(G)$ are
\begin{equation}
f(U)=\mathrm{Tr}[A (U \otimes \overline{U})^{\otimes t}], g \in G   
\end{equation}
so they are balanced polynomials in $U$ and $\overline{U}$ of degree $t$. Thus increasing $t$ corresponds to considering polynomials with higher degrees. This motivates us to consider the following function spaces in $L^2(G_d)$,
\begin{equation}
L^2_t(G_d)=\bigoplus_{\lambda \in \Lambda_t}V_{\lambda} \mathrm{,} 
\end{equation}
where $\Lambda_t$ is the set of unique (i.e. without repetitions and up to isomorphism) highest weights of irreps of $G_d$ appearing in $(U \otimes \overline{U})^{\otimes t}$. In the case $t=0$, we set
\begin{equation}
L^2_0(G_d) = I \mathrm{.}   
\end{equation}

Additionally we define the following related symbols. The set $\tilde{\Lambda}_t$ which equals $\Lambda_t$ without the weight of the trivial representation and the set of all unique highest weights
\begin{equation}
\Lambda_{\infty} \coloneqq \bigcup_{t=0}^{\infty} \Lambda_t \mathrm{.}    
\end{equation}
Fortunately the weights $\Lambda_t$ have a nice description in terms of the sequences of integers.

\begin{lemma}
\label{le:weights}
The set $\Lambda_t$ consists precisely of weights indexed by nonincreasing length $d$ integer sequences $\lambda$ such that $|\lambda|=0$ and $|\lambda_{+}| \leq t$, where $|\lambda|$ denotes the sum of entries and $\lambda_+$ is the subsequence of positive entries.
\end{lemma}


Each sequence $\lambda=(\lambda_1, \ldots, \lambda_d) \in \Lambda_t$ corresponds to a weight (the linear functional on the Cartan subalgebra $\mathfrak{h} \subset \mathfrak{su}(d)$)
\begin{equation}
\lambda=\lambda_1 L_1 + \ldots \lambda_d L_d \mathrm{,} 
\end{equation}
where $L_i$ are the standard basis elements \footnote{The linear functional $L_i$ returns the $i$-th diagonal entry of a matrix in $\mathfrak{h}$.}. Since $L_1 + \ldots + L_d=0$ in $\mathfrak{h}^{*}$, adding a constant sequence, $(c, \ldots, c)$ for some $c \in \mathbb{Z}$, to $\lambda$ does not change the weight.
\begin{example}
Consider the system of two qutrits $\mathbb{C}^3 \otimes \mathbb{C}^3$ and $t=2$. Then, from Lemma \ref{le:weights}, we have
\begin{equation}
\Lambda_t=\{(2, 0, -2), (2, -1, -1), (1, 1, -2), (1, 0 ,-1), (0, 0 ,0)\}   
\end{equation}
which is equivalent to
\begin{equation}
\Lambda_t=\{(4, 2, 0), (3, 0, 0), (3, 3, 0), (2, 1 ,0), (0, 0 ,0)\} \mathrm{,}    
\end{equation}
and for example $\lambda = (2, 1, 0)$ corresponds to the highest weight $2 L_1 + L_2$ i.e. to the adjoint representation. Similarly we can represent $\tilde{\Lambda}_t$ as
\begin{equation}
\tilde{\Lambda}_t=\{(4, 2, 0), (3, 0, 0), (3, 3, 0), (2, 1 ,0)\} \mathrm{.}   
\end{equation}
\end{example}
We introduce the following norm on the space of weights of $G_d$
\begin{equation}
    ||\lambda||_1 \coloneqq \sum_{i=1}^d |\lambda_i| \mathrm{.}
\end{equation}
 It is clear that for each $\lambda \in \Lambda_t$,
\begin{equation}
\label{eq:bn}
||\lambda||_1 \leq 2t \mathrm{.}  
\end{equation}
From now on we represent each irrep $\lambda$ by the sequence with smallest $||\lambda||_1$. In particular, the trivial representation is given by $\lambda=(0,0,\ldots, 0)$.

By choosing the orthonormal basis of function spaces (\ref{eq:fbasis}) we have the isomorphisms
\begin{equation}
V_{\lambda} \simeq \mathcal{H}_{\lambda} \mathrm{,}
\end{equation} where $\mathcal{H}_{\lambda} \coloneqq \mathbb{C}^{d_{\lambda}}$. We define \begin{equation}
 \mathcal{H}_t \coloneqq \bigoplus_{\lambda \in \Lambda_t}\mathcal{H}_{\lambda} \simeq L^2_t(G_d) \mathrm{,}    
\end{equation}
and analogously we define $\mathcal{H}_{\infty}$ \footnote{Here we use the closure of the direct sum.}. By $\mathcal{H}$ we denote the vector space isomorphic to $L^2(G_d)$.
Clearly,
\begin{equation}
 L^2(G_d) \simeq \mathcal{H} \coloneqq \bigoplus_{\lambda \in \Lambda_{\infty}}\mathcal{H}_{\lambda}^{\oplus d_{\lambda}} \mathrm{.}   
\end{equation}

For any representation of $G_d$ and any finite Borel measure $\nu$ on $G_d$ we define the operator 
\begin{equation}
\label{eq:rhoavg}
\pi(\nu) \coloneqq \int_{G_d} \pi(g) d \nu(g) \mathrm{,}  
\end{equation}
acting on the representation space of $\pi$. We use can use (\ref{eq:rhoavg}) to define various averaging operators. By $\mathcal{S}$ we denote a finite set of generators of $G_d$ and $\nu_{\mathcal{S}}$ is the counting measure of $\mathcal{S}$ on $G_d$.

The $t$-averaging operator wrt to $\mathcal{S}$, 
$T_{\nu_{\mathcal{S}},t}: \mathcal{H}_t\rightarrow \mathcal{H}_t$ is
\begin{equation}
\label{eq:def1}
 T_{\nu_{\mathcal{S}},t} \coloneqq \bigoplus_{\lambda \in \Lambda_t}\pi_{\lambda}(\nu_S) \mathrm{,}
\end{equation}
and can be represented as a block-diagonal matrix.
Analogously we define the $\infty$-averaging operator wrt to $\mathcal{S}$,
$T_{\nu_{\mathcal{S}},\infty}: \mathcal{H}_{\infty}\rightarrow \mathcal{H}_{\infty}$,
\begin{equation}
 T_{\nu_{\mathcal{S}}, \infty} \coloneqq \bigoplus_{\lambda \in \Lambda_{\infty}}\pi_{\lambda}(\nu_S) \mathrm{.}
\end{equation}
Finally, the (global) averaging operator wrt to $\mathcal{S}$, $T_{\nu_{\mathcal{S}}}: \mathcal{H}\rightarrow \mathcal{H}$ is
\begin{equation}
 T_{\nu_{\mathcal{S}}} \coloneqq \bigoplus_{\lambda \in \Lambda_{\infty}}\pi_{\lambda}(\nu_S)^{\oplus d_{\lambda}} \mathrm{.}
\end{equation}
In the language of functions, introduced averaging operators correspond to restrictions of $\mathrm{Reg}(\nu_{\mathcal{S}})$ to corresponding function subspaces. We denote such isomorphic averaging operators using the same symbols. For example, the global averaging operator is
\begin{equation}
 T_{\nu_{\mathcal{S}}} = \mathrm{Reg}(\nu_{\mathcal{S}}) \mathrm{,}
 \end{equation}
 so the action on $f \in L^2(G_d)$ is
 \begin{equation}
 \label{eq:int}
 (T_{\nu_{\mathcal{S}}} f)(h)= \int_{G_d} (\mathrm{Reg} (g) f)(h) d \nu_{\mathcal{S}}(g) = \frac{1}{|\mathcal{S}|}  \sum_{i=1}^{|\mathcal{S}|} f(g_i^{-1} h) \mathrm{.}
\end{equation}

The justification for the name averaging operator is clear from (\ref{eq:int}). Indeed,  $T_{\nu_{\mathcal{S}}}$ replaces the function $f$ with the averaged function, whose value at $h$ is the average of the values of $f$ over all translates of $h$ by the elements of $\mathcal{S}$.

Similarly, the $t$-averaging operator is
\begin{equation}
T_{\nu_{\mathcal{S}},t} = \restr{\mathrm{Reg}(\nu_{\mathcal{S}}) }{L^2_t(G_d)} \mathrm{,}    
\end{equation} 
so it acts just like $T_{\nu_{\mathcal{S}}}$ but on a restricted domain of functions.


Since $T_{\nu_{\mathcal{S}},t}$ is a sum of $|\mathcal{S}|$ left shift operators, normalized by $1/|\mathcal{S}|$,
and due to left-invariance of Haar measure, each such operator is unitary on $L^2_t(G_d)$, we see that $||T_{\nu_{\mathcal{S}},t}||_{op} \leq 1$, where by $||\cdot||_{op}$ we denote the operator norm. On the other hand $T_{\nu_{\mathcal{S}},t}$ acts trivially on $\mathcal{H}_{\lambda_0}$, where $\lambda_0=(0,0,\ldots,0)$ so $||T_{\nu_{\mathcal{S}},t}||_{op} \geq 1$ and hence $||T_{\nu_{\mathcal{S}},t}||_{op}=1$.

The subspace $\mathcal{H}_{\lambda_0}$ corresponds to the subspace of constant functions $L^2_0(G_d)=V_{\lambda_0}$, with orthogonal compliment being the space of functions with Haar-average zero. Let $T_{\mu}$ denote the projector onto $\mathcal{H}_{\lambda_0}$. At the level of function spaces, $T_{\mu}$ is the projector onto $L^2_0(G_d)$ which assigns to each function $f$ the constant function with value being the Haar average \footnote{From now on the symbol $\mu$ denotes the Haar measure on $G_d$.} of $f$ ,
\begin{equation}
(T_{\mu}f)(h)=\int_{G_d} f (g) d \mu(g) \mathrm{.}    
\end{equation}

    

In order to assess how quick the words in $\mathcal{S}$ fill the group $G_d$, we compare the averaging operator $T_{\nu_{\mathcal{S}}}$ with $T_{\mu}$ by checking the operator of their difference. Since $\restr{T_{\nu_{\mathcal{S}}}}{\mathcal{H}_{\lambda_0}}=\restr{T_{\mu}}{\mathcal{H}_{\lambda_0}}$, the norm $||T_{\nu_{\mathcal{S}}} - T_{\mu}||_{op}$ equals the norm of the operator
\begin{equation}
\tilde{T}_{\nu_{\mathcal{S}}} \coloneqq  \bigoplus_{\lambda \in \tilde{\Lambda}_{\infty}}\pi_{\lambda}(\nu_S)^{\oplus d_{\lambda}} \mathrm{.}    
\end{equation}
Similarly, we define $\tilde{T}_{\nu_{\mathcal{S}}, t}$ and $\tilde{T}_{\nu_{\mathcal{S}}, \infty}$.
Clearly, 
\begin{equation}
  \sigma(T_{\nu_{\mathcal{S}}}) = \sigma(T_{\nu_{\mathcal{S}}, \infty}) \mathrm{,}  
\end{equation} so we have
\begin{equation}
 ||T_{\nu_{\mathcal{S}}} - T_{\mu}||_{op}=||\tilde{T}_{\nu_{\mathcal{S}}}||_{op}=||\tilde{T}_{\nu_{\mathcal{S}},\infty}||_{op} \mathrm{.}   
\end{equation}
This motivates us to define the spectral gap of $\mathcal{S}$ as
\begin{equation}
\label{eq:g3}
\mathrm{gap}(\mathcal{S}) \coloneqq  1-||\tilde{T}_{\nu_{\mathcal{S}}}||_{op} \mathrm{.}
\end{equation}

The spectral gap is an useful numerical value describing the set $\mathcal{S}$ via the properties of the corresponding averaging operator.

Similarly, we define spectral gap of $\mathcal{S}$ at scale $t$ as
\begin{equation}
\label{eq:g2}
\mathrm{gap_t}(\mathcal{S}) \coloneqq 1-||\tilde{T}_{\nu_{\mathcal{S}},t}||_{op} \mathrm{.}   
\end{equation}
In general we can define analogous gaps for any finite Borel measure $\nu$ on $G_d$. For example,
\begin{equation}
\label{eq:g1}
\mathrm{gap_t}(\nu) \coloneqq 1-||\tilde{T}_{\nu,t}||_{op} \mathrm{,}   
\end{equation}
where $\tilde{T}_{\nu,t}$ is defined as in (\ref{eq:def1}) with $\nu_{\mathcal{S}}$ substituted by $\nu$.
It is clear that
\begin{equation}
\mathrm{gap}(\mathcal{S}) =\inf_{t} \mathrm{gap_t}(\mathcal{S}) \mathrm{,}   
\end{equation}
and the gaps (\ref{eq:g3})-(\ref{eq:g1}) belong to $[0,1]$.

We argue that we can assume that $\mathcal{S}$ is symmetric without the loss of generality. For a measure $\nu$ on $G_d$ we define its conjugate $\tilde{\nu}$ via the property
\begin{equation}
    \int_{G_d} f(g) d\tilde{\nu}(g) =  \int_{G_d} f(g^{-1}) d\nu(g)
\end{equation}
for all continuous functions $f$ on $G_d$. We say a measure $\nu$ is symmetric if $\nu = \tilde{\nu}$. For two measures $\nu_1$ and $\nu_2$ on $G_d$, their convolution $\nu_1 * \nu_2$ is a measure on $G_d$ defined via
\begin{equation}
\nu_1 * \nu_2(\Omega) = \int_{G_d} \mathbbm{1}_\Omega(gh) \nu_1(g) \nu_2(h) \mathrm{.}
\end{equation}
Going back to the definition (\ref{eq:rhoavg}) we have
\begin{equation}
\pi(\nu_1 * \nu_2) = \pi(\nu_1) \pi(\nu_2) \mathrm{.}    
\end{equation}
It is easy to see that $\pi(\tilde{\nu})=\pi(\nu)^*$. In particular if $\nu$ is symmetric then $\pi(\nu)$ is self-adjoint and hence $\sigma(\pi(\nu))$ is real. Note also that $\nu * \tilde{\nu}$ is automatically symmetric. We can write

\begin{equation}
\pi(\tilde{\nu} * \nu)=\pi(\tilde{\nu}) \cdot \pi(\nu)=\pi(\nu)^* \cdot \pi(\nu) \mathrm{,}   
\end{equation}
which means that
\begin{equation}
||\pi(\tilde{\nu} * \nu)||_{op}=||\pi(\nu)||_{op}^2 \mathrm{.}    
\end{equation}
Finally, because $\sqrt{1-x} \leq 1-\frac{x}{2}$ for any $0 \leq x \leq 1$,
\begin{equation}
\mathrm{gap}_t( \tilde{\nu}_{\mathcal{S}} * \nu_{\mathcal{S}}) \geq \mathrm{gap}_t(\mathcal{S}) \geq \frac{1}{2}\mathrm{gap}_t( \tilde{\nu}_{\mathcal{S}} * \nu_{\mathcal{S}}) \mathrm{.}   
\end{equation}

Since $\mathcal{S}$ is symmetric, $T_{\nu_{\mathcal{S}},t}$ is Hermitian and so its spectrum $\sigma(T_{\nu_{\mathcal{S}},t})$ is contained in $[-1,1]$. The same is true for $T_{\nu_{\mathrm{S}},\infty}$ and  $T_{\nu_{\mathrm{S}}}$. Note that since the subspace $\mathcal{H}_{\lambda_0}$ is excluded, the question if $\mathrm{gap}(\mathcal{S}) > 0$ is non-trivial. The gap exists, i.e.  $\mathrm{gap}(\mathcal{S}) > 0$, if and only if $1$ belongs to the spectrum $\sigma(T_{\nu_{\mathcal{S}}})$ i.e. it is the accumulation point of $\sigma(T_{\nu_{\mathcal{S}}})$. In such a case we say that $T_{\nu_{\mathcal{S}}}$ has a spectral gap.


Let's denote by $\mathcal{S_{\ell}}$ a set of words in $G_d$ of length $\ell$ built from elements of $\mathcal{S}$
\begin{equation}
 \mathcal{S_{\ell}} \coloneqq \{g_1 g_2 \ldots g_{\ell} |\, g_1, g_2, \ldots g_\ell \in \mathcal{S}\} \mathrm{.}
 \end{equation}
 The corresponding averaging operator is $T_{\nu_{\mathcal{S}}}^{\ell}$. Indeed,
\begin{equation}
 T_{\nu_{\mathcal{S}}}^{\ell}=T_{\nu_{\mathcal{S}}^{*(\ell)}} \mathrm{,}
\end{equation}
and since $\nu_{\mathcal{S}}^{*(\ell)}$ is the law for $\mathcal{S}_{\ell}$, $T_{\nu_{\mathcal{S}}}^{\ell}$ is the averaging operator with respect to $\mathcal{S}_{\ell}$.
At the level of functions we have
\begin{equation}
(T_{\nu_{\mathcal{S}}}^{\ell} f)(h)=\frac{1}{|\mathcal{S}|^{\ell}}  \sum_{w\in S_{\ell}}  f(w^{-1}g) \mathrm{.}   
\end{equation}
Importantly, $\mathrm{gap}(\mathcal{S})$ can be interpreted as the exponential rate of convergence of the (global) averaging operator $T_{\nu_{\mathcal{S}}}$ to $T_{\mu}$ in the operator norm with $\ell$ increasing. Indeed, due to left-invariance of Haar measure
\begin{equation}
 T_{\nu_{\mathcal{S}}} T_{\mu}=T_{\mu}=T_{\mu} T_{\nu_{\mathcal{S}}} \mathrm{,}   
\end{equation} so we have
\begin{equation}
 ||T_{\nu_{\mathcal{S}}}^{\ell}-T_{\mu}||_{op}= ||(T_{\nu_{\mathcal{S}}}-T_{\mu})^{\ell}||_{op} \leq ||(T_{\nu_{\mathcal{S}}}-T_{\mu})||_{op}^{\ell}=||\tilde{T}_{\nu_{\mathcal{S}}}||_{op}^{\ell}\mathrm{,}   
\end{equation}
and using the notion of a spectral gap (\ref{eq:g3}) we have
\begin{equation}
 ||T_{\nu_{\mathcal{S}}}^{\ell}-T_{\mu}||_{op} \leq (1-\mathrm{gap}(\mathcal{S}))^{\ell}\leq e^{- \ell \, \mathrm{gap}(\mathcal{S})}  \mathrm{.}   
\end{equation}


Thus, if the gap exists then $T_{\nu_{\mathcal{S}}}^{\ell}$ converges to $T_{\mu}$ as $\ell \to \infty$ exponentially fast in the operator norm. Moreover, the rate of convergence improves exponentially with $\mathrm{gap}(\mathcal{S})$ increasing. This motivates us to study $\mathrm{gap}(\mathcal{S})$.
\section{Bound on gates efficiency from spectral gap}
\label{sec:bog}

In \cite{harrow02} it has been shown in case of group $\mathrm{SU}(d)$ that if $\mathcal{S}$ is universal and $T_{\nu_{\mathcal{S}}}$ has a spectral gap then words of length $\ell = \mathcal{O}\left(\mathrm{log}\left( \frac{1}{\epsilon}\right)\right)$ form an $\epsilon$-net.

In this section we present an alternative proof of this fact for $G_d$. By $D(\cdot, \cdot)$ we denote a $G_d$-invariant metric on $G_d$ defined as follows. For $g, h \in G_d$ let $\hat{g}, \hat{h} \in \mathrm{U}(d)$ be the corresponding representatives. Then
\begin{equation}
 D(g, h) \coloneqq \mathrm{inf}_{\theta \in [0, 2 \pi)} ||e^{i \theta}\hat{g} - \hat{h}||_{op}\mathrm{.}  
\end{equation}
  Equivalently, we can take the infimum over representatives
\begin{equation}
D(g, h) = \mathrm{inf}_{\hat{g}, \hat{h}} ||\hat{g} - \hat{h}||_{op} \mathrm{.}  
\end{equation}

We introduce $B(x,r)$ as the closed ball in $G_d$ with radius $r$ centered at $x$ with respect to $D$ and $B(r)$ is such ball centered at $\mathbb{I}$.

By $\mathrm{Vol}(\Omega)$ we mean the Haar volume of a subset $\Omega \subset G_d$,
\begin{equation}
\mathrm{Vol}(\Omega) = \int_{G_d} \mathbbm{1}_{\Omega}(g) d \mu(g) \mathrm{,} 
\end{equation}
where $\mathbbm{1}_{\Omega}$ denotes the indicator function of $\Omega$.

We start with the following simple observation. Let $f \in L^2_t(G_d)$, $\int_{G_d} f(g) d\mu(g)=1$, and pick some region $\Omega \subseteq G_d$. Since
\begin{equation}
\mathrm{Vol}({\Omega})=\int_{\Omega} 1 d \mu(g)
\end{equation} we have that
\begin{equation}
\int_{\Omega} 1 d \mu(g)-\int_{\Omega} (T_{\nu_{\mathcal{S}},t}^{\ell} f)(g) d \mu(g)=\langle 1-T_{\nu_{\mathcal{S}},t}^{\ell}f | \mathbbm{1}_{\Omega} \rangle \leq ||1-T_{\nu_{\mathcal{S}},t}^{\ell}f||_{2} \cdot \sqrt{\mathrm{Vol}({\Omega})}    
\end{equation}
and
\begin{equation}
||1-T_{\nu_{\mathcal{S}},t}^{\ell}f||_{2}=||(T_{\nu_{\mathcal{S}},t}^{\ell}-T_{\mu})f||_{2} \leq ||(T_{\nu_{\mathcal{S}},t}^{\ell}-T_{\mu})||_{op} \cdot||f||_{2}  \leq e^{-\ell \, \mathrm{gap}_t({\mathcal{S}})} \cdot||f||_2 \mathrm{,}   
\end{equation}
where $T_{\mu}$ is a projector onto $L^2_0(G_d)$ on $L^2_t(G_d)$. Thus,
\begin{equation}
\label{obs:balls}
 \int_{\Omega} \left(T_{\nu_{\mathcal{S}},t}^{\ell} f \right)(g) \geq \mathrm{Vol}({\Omega})-e^{-\ell \, \mathrm{gap}_{t}(\mathcal{S})}||f||_2 \sqrt{\mathrm{Vol}({\Omega})}  \mathrm{.} 
\end{equation}
Clearly, analogous results are true for other averaging operators, in particular for $T_{\nu_{\mathrm{S}}}$.

\begin{theorem}
\label{th:ahr}
Assume $\mathcal{S}$ is such that $T_{\mathcal{S}}$ has a spectral gap. Then $S_{\ell}$ is an $\epsilon$-net for every $\ell$
$$\ell \geq \frac{\mathrm{dim}\,G_d}{\mathrm{gap}(\mathcal{S})} \mathrm{log}\left(\frac{1}{\epsilon}\right)+B \mathrm{.}$$
\end{theorem}
\begin{proof}
Pick an element $U_0 \in G_d$ and a ball ${\Omega}=B(U_0,\epsilon/2)$ centered at it. Pick $\ell$ such that there is no $w_{\ell}\in \mathcal{S}_{\ell}$ which $\epsilon$-approximates $U_0$, i.e. such that $D(w_{\ell},U_0) \leq \epsilon$.
Let $f$ be a normalized indicator function of $B(\mathbb{I},\epsilon/2)$, i.e.
\begin{equation}
\label{eq:nif}
 f(x)=\frac{1}{\mathrm{Vol}(B(\mathbb{I},\epsilon/2))}\mathbbm{1}_{B(\mathbb{I},\epsilon/2)}(x) \mathrm{.}
\end{equation}
We have
\begin{equation}
\int_{\Omega} \left(T_{\nu_{\mathcal{S}}}^{\ell} f \right)(g) d\mu(g)=\frac{1}{|\mathcal{S}|} \sum_{w_{\ell} \in \mathcal{S}^{\ell}} \int_{G_d} f(w_{\ell}^{-1} g) d \mu(g)    
\end{equation}
but for each $g$ in ${\Omega}$
\begin{equation}
D(w_{\ell}^{-1}g ,\mathbb{I})=D(g, w_{\ell}) > \epsilon/2 \mathrm{,}    
\end{equation}
hence
\begin{equation}
\label{eq:intz}
\int_{\Omega} \left(T_{\nu_{\mathcal{S}}}^{\ell} f \right)(g) d \mu(g)=0    
\end{equation}
Using (\ref{obs:balls}) we get
\begin{equation}
e^{-\ell \, \mathrm{gap}(\mathcal{S})} \geq \mathrm{Vol}({\Omega}) \mathrm{,}    
\end{equation}
since $||f||_2=1/\sqrt{\mathrm{Vol}({\Omega})}$.
Hence, if
\begin{equation}
 e^{-\ell \, \mathrm{gap}(\mathcal{S})} < \mathrm{Vol}({\Omega}) \mathrm{,}   
\end{equation}
we get a contradiction, which means that
$\mathcal{S}_{\ell}$ is an $\epsilon$-net. On the other hand
\begin{equation}
\mathrm{Vol}({\Omega}) \leq C_V (\epsilon/2)^{\mathrm{dim}\, G_d} \mathrm{,}    
\end{equation}
where $C_V$ is some group constant. Thus,
\begin{equation}
\ell \geq \frac{\mathrm{dim}\, G_d}{\mathrm{gap}(\mathcal{S})} \mathrm{log}\left(\frac{1}{\epsilon}\right)+B \mathrm{,}    
\end{equation}
with
\begin{equation}
 B=-\frac{\mathrm{log}(C_V)-\mathrm{dim}\, G_d \cdot \mathrm{log}(2)}{\mathrm{gap}(\mathcal{S})} \mathrm{.}   
\end{equation}
\end{proof}
\noindent
We have $\mathrm{dim}\, G_d=d^2-1$ and in the case of $G_d$ can put $C_V=(9.5)^{d^2-1}$, so
\begin{equation}
B = - \frac{d^2-1}{\mathrm{gap}(\mathcal{S})}\mathrm{log}(4.75) \mathrm{.}   
\end{equation}
The values of constant $C_V$ bounding the volume of a ball in various groups can be obtained by techniques from \cite{szarek1998}.


Note that Theorem \ref{th:ahr} cannot be stated in analogous form for the $t$-averaging operators $T_{\nu_{\mathcal{S}}, t}$, since the normalized indicator function (\ref{eq:nif}) does not belong to $L^2_t(G_d)$ for any $t$ so we cannot write (\ref{eq:intz}) for $T_{\nu_{\mathcal{S}},t}$ instead of $T_{\nu_{\mathcal{S}}}$. However, by considering appropriate approximations of Dirac delta by polynomials from $L^2_t(G_d)$, we can show that
\begin{equation}
\int_{\Omega} \left(T_{\nu_{\mathcal{S}},t}^{\ell} f \right)(g) d \mu(g)   
\end{equation}
is sufficiently small and hence obtain analogous results. In particular, it is known that 
\begin{equation}
\ell \geq \frac{C \cdot \mathrm{log}\left(\frac{1}{\epsilon}\right)}{\mathrm{gap}_r (\mathcal{S})} \mathrm{,}
\end{equation}
where $r=D/ \epsilon^{2 (d^2-1)+2}$ and $C, D$ are some constants \cite{varju13}.
This result has been improved in case of $\mathrm{U}(d)$ in \cite{oszmaniec22}, where
\begin{equation}
\label{eq:bmo}
\ell \geq \frac{(d^2-1)(2 \mathrm{log}\left(\frac{1}{\epsilon}\right) + \mathrm{log}(4 C_b^{3/2}d)) + \mathrm{log}(32)}{\mathrm{gap}_t(\mathcal{S})} \mathrm{,}    
\end{equation}
for some absolute constant $C_b$ and $t \geq 5 d^{5/2} /\epsilon \cdot \tau(\epsilon,d)$, where  $\tau(\epsilon,d)$ is

\begin{equation}
\label{eq:tau}
 \tau(\epsilon, d) = \mathrm{log}^{\frac{1}{2}} (6 C_b / \epsilon) \cdot \sqrt{\frac{1}{32} \mathrm{log}^{\frac{1}{2}} (6 C_b / \epsilon) +  \mathrm{log} \left( \frac{d}{\epsilon} \cdot \mathrm{log}^{\frac{1}{2}} (6 C_b / \epsilon)  \right)}  \mathrm{.}
\end{equation}

\section{Calculable lower bound on spectral gap}
\label{sec:mth}
In this section, we derive lower bounds on the spectral gap at scale $t$ for $\mathcal{S} \subset G_d$, such that any two pairs in $\mathcal{S}$ (of gate with its inverse) form an universal set themselves. This condition can be verified numerically by known universality criteria, see e.g. \cite{ sawicki21}.

Our bound for any $t$ can be calculated from the knowledge of certain gaps up to some fixed $t_0=t_0(d)$ and is of the form

\begin{equation}
\mathrm{gap}_t(\mathcal{S}) \geq  \alpha \cdot  g_{t_0}(\mathcal{S}) \cdot   \mathrm{log}(\beta t)^{-2c} \mathrm{,}
\end{equation}
where $\alpha, \beta, c > 0$ are some specific calculable constants and $g_{t_0}(\mathcal{S})$ can be determined numerically by calculating gaps of certain sets derived from $\mathcal{S}$ up to some calculable scale $t_0$.

We study the action of the $t$-averaging operator wrt to $\mathcal{S}$, 
\begin{equation}
 T_{\nu_{\mathcal{S}},t} \coloneqq \bigoplus_{\lambda \in \Lambda_t}\pi_{\lambda}(\nu_S) \mathrm{,}
\end{equation}
acting on the Hilbert space \begin{equation}
\mathcal{H}_t=\bigoplus_{\lambda \in \Lambda_t} \mathcal{H}_{\lambda} \mathrm{.}    
\end{equation}

\noindent
By $S(\mathcal{H}_{\lambda})$ we denote the unit sphere in $\mathcal{H}_{\lambda}$,
\begin{equation}
S(\mathcal{H}_{\lambda})=\{w \in \mathcal{H}_{\lambda} |\, ||w||=1\} \mathrm{.}    
\end{equation}
We choose the orthonormal basis
\begin{equation}
\{w^{\lambda}_{ij} | \,1 \leq i, j \leq d_{\lambda}, \, \lambda \in \Lambda_t \}    
\end{equation} of $\mathcal{H}_t$, induced by the basis (\ref{eq:fbasis}).
Clearly, $||T_{\nu_{\mathcal{S},t}}||_{op} \leq 1$ and our goal is to improve this bound.
The irreps $\Lambda_t$ of $G_d$ can be divided into three disjoint sets, based on the type of the representation of $\mathrm{U}(d)$ they come from:
\begin{equation}
    \Lambda_t = \Lambda_{t, \mathbb{H}} \cap \Lambda_{t, \mathbb{R}} \cap \Lambda_{t, \mathbb{C}} \mathrm{,}
\end{equation}
where $\mathbb{H}$, $\mathbb{R}$ and $\mathbb{C}$ stands for quaternionic, real and complex representations. In fact, $\Lambda_{t, \mathbb{H}}=\emptyset$ since quaternionic representations of $\mathrm{U}(d)$ do not contribute to projective representations.

Since
\begin{equation}
|| \tilde{T}_{\nu_{\mathcal{S}},t}||_{op}=\mathrm{max}_{\lambda \in \tilde{\Lambda}_t} ||\pi_{\lambda}(\nu_{\mathcal{S}})||_{op} \mathrm{,}    
\end{equation}
we fix any $\lambda \in \tilde{\Lambda}_t$ and consider $||\pi_{\lambda}(\nu_{\mathcal{S}})||_{op}$.\\

Additionally we assume $\mathcal{S}=\{U_1, \ldots, U_k, U_1^{-1}, \ldots U_k^{-1}\}$ is generic so that for each $1 \leq i \neq j \leq k$, the set $\{U_i^2, U_j^2, U_i^{-2}, U_j^{-2}\}$ is a universal symmetric set.\\
\textbf{Strategy of the proof}\\
Our strategy is to show that for any $\lambda \in \tilde{\Lambda}_{t}$, any $w \in S(\mathcal{H}_{\lambda})$ and any generator $U_m \in \mathcal{S}$, except for at most one, say $U_k$,
\begin{equation}
||(\pi_{\lambda}(U_i)+\pi_{\lambda}(U_i^{-1}))w|| \leq 2-(b_i/2)^2    
\end{equation} for some coefficients $b_i^2=b_i^2(\lambda)>0$ which can be bounded by gaps of certain subsets of the set $\mathcal{S}^2 = \{U_1^2, \ldots, U_k^2, U_1^{-2}, \ldots U_k^{-2}\}$ at some known scale $t_0$. Hence,
\begin{equation}
||\pi_{\lambda}(\nu_{\mathcal{S}}) w|| \leq \frac{1}{|\mathcal{S}|} \left[\sum_{1 \leq i < k} ||(\pi_{\lambda} (U_i)  +\pi_{\lambda} (U_i^{-1}))w ||+2\right] \leq 1 -  \frac{1}{4|\mathcal{S}|} \sum_{i=1}^{|\mathcal{S}|/2-1} b_i^2(\lambda) \mathrm{,}    
\end{equation}
which implies 
\begin{equation}
\mathrm{gap}_t(\mathcal{S}) \geq \mathrm{min}_{\lambda \in \tilde{\Lambda}_{t}} \frac{1}{4 |\mathcal{S}|} \sum_{i=1}^{|\mathcal{S}|/2-1} b_i^2(\lambda) > 0  \mathrm{.}    
\end{equation}

This means that we can obtain a non-trivial lower bound on $\mathrm{gap}_t(S)$ for any $t \geq t_0$. Crucially, the value of $t_0$ can be easily determined and is not large, so the numerical calculations of the bound are feasible.\\
\noindent
\textbf{The main reasoning}\\
Since $\pi_{\lambda}(g_i)$ is unitary we have
\begin{equation}
\label{eq:pm}
||(\pi_{\lambda}(g_i)+\pi_{\lambda}(g_i^{-1}))w||^2=4-||(\pi_{\lambda}(g_i)-\pi_{\lambda}(g_i^{-1}))w||^2   
\end{equation}
for any $w \in S(\mathcal{H}_{\lambda})$.
Let $\imath_{\lambda}$ denote the Frobenius-Schur indicator of $\pi_\lambda$,
\begin{equation}
\imath_{\lambda}=\int_{G_d} \chi_{\lambda}(g^2) d \mu(g) = \begin{cases}
    -1,  \quad \mathrm{if}\, \lambda \in \Lambda_{t, \mathbb{H}},\\
    0, \quad \mathrm{if}\, \lambda \in \Lambda_{t, \mathbb{C}} \mathrm{,}\\
    1, \quad \mathrm{if}\, \lambda \in \Lambda_{t, \mathbb{R}} \mathrm{.}
    \end{cases} \mathrm{.}    
\end{equation}
Note that
\begin{equation}
\int_{G_d} \pi_{\lambda}(g^2) d \mu(g) = \frac{\imath_{\lambda}}{d_{\lambda}} \mathbb{I_{\lambda}}     
\end{equation} since the LHS is a self-intertwiner.


\noindent
Observe that since $||\lambda||_1 > 0$, for any $i,j$ and $p,q$ we have
\begin{equation}
\int_{G} \langle \pi_{\lambda}(g^2) w_{ij}^{\lambda}, w_{pq}^{\lambda} \rangle d\mu(g) =  \frac{\imath_{\lambda}}{d_{\lambda}} ||w_{ij}^{\lambda}||^2 \delta_{ip} \delta_{jq} \mathcal{.}   
\end{equation}
\noindent
Hence, for any $w\in S(\mathcal{H}_{\lambda})$
\begin{equation}
\int_{G} \langle \pi_{\lambda}(g^2) w, w \rangle d\mu(g) =  \frac{\mathrm{\imath_{\lambda}}}{d_{\lambda}} \mathrm{,}    
\end{equation}
so for any $\lambda \in \Lambda_t$, there exists $h=h(w) \in G_d$ such that
$ \mathfrak{Re}\langle \pi_{\lambda}(h^2) w, w \rangle < \frac{\imath_{\lambda}}{d_{\lambda}} $ and
\begin{equation}
\label{eq:lnorm}
||(\pi_{\lambda}(h)-\pi_{\lambda}(h^{-1}))w|| > \sqrt{2\left( 1 - \frac{\imath_{\lambda}}{d_{\lambda}}\right)} \mathrm{.}    
\end{equation}
Note that if $\lambda$ is quaternionic the bound is even better.

We want to connect $h^2$ with a square of some generator $g_i^2$, so that the large value of the norm (\ref{eq:lnorm}) will propagate to the large value of $||(\pi_{\lambda}(g_i^2)-\pi_{\lambda}(e))w||$.
Let
\begin{equation}
\mathcal{S}^2 \coloneqq \{U_1^2, \ldots, U_k^2, U_1^{-2}, \ldots U_k^{-2}\}   
\end{equation}
and by $\mathcal{S}^2_{j_1 \ldots j_m}$ we denote the set $\mathcal{S}^2$ without elements $U^2_{j_1}, \ldots, U^2_{j_m}$ (and their inverses),
\begin{equation}
\label{eq:sind}
\mathcal{S}^2_{j_1 \ldots j_m} \coloneqq \mathcal{S}^2 \setminus   \{U^2_{j_1},U^{-2}_{j_1} \ldots, U^2_{j_m}, U^{-2}_{j_m}\} \mathrm{.}
\end{equation}
By the assumption, each set $\mathcal{S}^2_{j_1 \ldots j_m}$ is universal for $1 \leq m \leq k-2$.
Consider any $\mathcal{S}^2_{j_1 \ldots j_m}$ (we allow $\mathcal{S}^2$ as the special case with $m=0$).
We find an $\epsilon_{j_1 \ldots j_m}$-approximation of $h^2$ in terms of squares generators, namely we write $\tilde{h}=g_1^2 g_2^2 \ldots g_{\ell_{j_1 \ldots j_m}}^2$ , where each $g_i^2 \in \mathcal{S}^2_{j_1 \ldots j_m}$, so that $D(h^2,\tilde{h}) < \epsilon$ and we specify $1 > \epsilon_{j_1 \ldots j_m} >0$ later. We have 
\begin{equation}
\sqrt{2\left( 1 - \frac{\imath_{\lambda}}{d_{\lambda}}\right)} < ||(\pi_{\lambda}(h^2)-\pi_{\lambda}(e))w|| \leq ||(\pi_{\lambda}(e)-\pi_{\lambda}(\tilde{h}))w||+||(\pi_{\lambda}(\tilde{h})-\pi_{\lambda}(h^2))w||   
\end{equation}

so
\begin{align}
||(\pi_{\lambda}(e)-\pi_{\lambda}(\tilde{h}))w|| &\geq \sqrt{2\left( 1 - \frac{\imath_{\lambda}}{d_{\lambda}}\right)} - ||(\pi_{\lambda}(\tilde{h})-\pi_{\lambda}(h^2))w||\\ &\geq \sqrt{2\left( 1 - \frac{\imath_{\lambda}}{d_{\lambda}}\right)} - ||\pi_{\lambda}(\tilde{h})-\pi_{\lambda}(h^2)||_{op} \nonumber \mathrm{.}    
\end{align}
From the unitary invariance of operator norm
\begin{equation}
||\pi_{\lambda}(\tilde{h})-\pi_{\lambda}(h^2)||_{op}=||\pi_{\lambda}(e)-\pi_{\lambda}(\bar{h})||_{op}\mathrm{,}   
\end{equation}
where $\bar{h}=\tilde{h}^{-1}h^2$
and $D(e,\bar{h}) < \epsilon_{j_1 \ldots j_m}$.
Let us fix a maximal torus $T \subset G_d$ with Lie algebra $\mathfrak{t} \subset \mathfrak{su}(d)$. We can write $\bar{h}=gtg^{-1}$, where $t \in T$ and $g \in G$. Clearly,
\begin{equation}
\mathrm{Spec}(\pi_{\lambda}(\bar{h}))=\mathrm{Spec}(\pi_{\lambda}(t)) \mathrm{.}    
\end{equation}
Let $\{e^{ i  \gamma_1}, \ldots, e^{  i
\gamma_{d_{\lambda}}}\}$ be the spectrum of $\pi_{\lambda}(t)$ and $\{w_1, \ldots, w_{d_{\lambda}}\}$ be an orthonormal basis of $\mathcal{H}_{\lambda}$ in which
\begin{equation}
\pi_{\lambda}(t) w_j = e^{ i \gamma_j} w_j   
\end{equation}
for $1 \leq j \leq d_{\lambda}$.
By the definition of a real weight we have
\begin{equation}
|\gamma_j|=|\langle\mu_j, H \rangle|=|(H_{\mu_j}, H)| \leq ||\lambda||_1 \cdot \mathrm{max}_i |\theta_i|   
\end{equation}
for some weight $\mu_j$ of irrep $\pi_{\lambda}$ and $H=\mathrm{log}(t)=\mathrm{diag}(i \theta_1, \ldots,  i\theta_{d}) \in \mathfrak{t}$. \footnote{ Note that $D(e,\bar{h})=D(e,t) < \epsilon_{j_1 \ldots j_m} < 1$ so $\mathrm{log}(t)$ exists.} We assume $\theta_i \in (-\pi, \pi]$ for each $i$.
Since $D(e,t) < \epsilon_{j_1 \ldots j_m}$ we have
\begin{equation}
|\theta_i|/\pi \leq  \;  |\mathrm{sin}(\theta_i/2)| < \epsilon_{j_1 \ldots j_m}/2    
\end{equation}
for each $i$, so
\begin{equation}
|\gamma_j| \leq  \pi ||\lambda||_1 \epsilon_{j_1 \ldots j_m}/2   \mathrm{.}
\end{equation}
Finally,
\begin{equation}
||\pi_{\lambda}(e)-\pi_{\lambda}(\bar{h})||_{op}=||\pi_{\lambda}(e)-\pi_{\lambda}(t)||_{op} \leq 2 \,\mathrm{max}_i \, |\mathrm{sin(\gamma_i/2)}| \leq \mathrm{max}_i \, |\gamma_i|
\end{equation}
hence 



\begin{equation}
||\pi_{\lambda}(e)-\pi_{\lambda}(\bar{h})||_{op} \leq C ||\lambda||_1 \cdot \epsilon_{j_1 \ldots j_m} \mathrm{,}   
\end{equation}
where $C=\pi/2$. Thus,
\begin{equation}
||(\pi_{\lambda}(e)-\pi_{\lambda}(\tilde{h}))w|| \geq \sqrt{2\left( 1 - \frac{\imath_{\lambda}}{d_{\lambda}}\right)}-C||\lambda||_1 \cdot \epsilon_{j_1 \ldots j_m} \mathrm{.}    
\end{equation}

We use triangle inequality to propagate the result into some generator.
\begin{align}
&\sqrt{2\left( 1 - \frac{\imath_{\lambda}}{d_{\lambda}}\right)}-C||\lambda||_1 \epsilon_{j_1 \ldots j_m} \leq ||(\pi_{\lambda}(e)-\pi_{\lambda}(\tilde{h}))w||\\ &=||(\pi_{\lambda}(e)-\pi_{\lambda}(g_1^2 g_2^2 \ldots g_{\ell}^2))w|| \leq ||(\pi_{\lambda}(e)-\pi_{\lambda}(g_1^2)))w||+||(\pi_{\lambda}(g_1^2)-\pi_{\lambda}(g_1^2 g_2^2)))w||+\ldots \nonumber\\ &+ ||(\pi_{\lambda}(g_1^2 g_2^2 \ldots g_{i-1}^2)-\pi_{\lambda}(g_1^2 g_2^2 \ldots g_i^2)))w||+\ldots+||(\pi_{\lambda}(g_1^2 g_2^2 \ldots g_{\ell-1}^2)-\pi_{\lambda}(g_1^2 g_2^2 \ldots g_{\ell}^2)))w|| \nonumber
\end{align}
so there exists $i$ such that
\begin{equation}
||(\pi_{\lambda}(g_1^2 g_2^2 \ldots g_{i-1}^2)-\pi_{\lambda}(g_1^2 g_2^2 \ldots g_i^2)))w|| \geq b_{j_1 \ldots j_m}(\lambda) \mathrm{,}    
\end{equation}
where
\begin{equation}
\label{eq:b}
b_{j_1 \ldots j_m}(\lambda)=\frac{\sqrt{2\left( 1 - \frac{\imath_{\lambda}}{d_{\lambda}}\right)}-C||\lambda||_1 \epsilon_{j_1 \ldots j_m}}{\ell_{j_1 \ldots j_m}}    
\end{equation}
and from the unitary invariance of operator norm
\begin{equation}
||(\pi_{\lambda}(g_1^2 g_2^2 \ldots g_{i-1}^2)-\pi_{\lambda}(g_1^2 g_2^2 \ldots g_i^2)))w||=||(\pi_{\lambda}(e)-\pi_{\lambda}( g_i^2)))w||    
\end{equation}
Since $g_i^2$ is an element from $\mathcal{S}^2_{j_1 \ldots j_m}$ and
\begin{equation}
||(\pi_{\lambda}(e)-\pi_{\lambda}(g_i^2))w||=||(\pi_{\lambda}(e)-\pi_{\lambda}(g_i^{-2}))w|| \mathrm{,}
\end{equation}
there exists $i_q$, where $1 \leq q \leq m$, such that
\begin{equation}
||(\pi_{\lambda}(e)-\pi_{\lambda}(U_{i_q}^2))w||\geq b_{j_1 \ldots j_m}(\lambda) \geq b_m(\lambda) \mathrm{,}    
\end{equation}
where
\begin{equation}
b_m(\lambda) \coloneqq \mathrm{min}_{i_1, \ldots, i_m} b_{i_1 \ldots i_m}(\lambda)    
\end{equation} is the bound for the worst choice of $i_1, \ldots, i_m$, which we denote $\mathcal{S}^2_m$ \footnote{Note that $b_m(\lambda) \geq b_{n}(\lambda)$ for $m < n$.}. The set $\mathcal{S}^2_m$ has the corresponding $\epsilon_m$ and $\ell_m$ via (\ref{eq:b}).

We proceed as follows. First, we consider above procedure for $\mathcal{S}^2$ and obtain
\begin{equation}
||(\pi_{\lambda}(e)-\pi_{\lambda}(U_{i_1}^2))w||\geq  b_0(\lambda) \mathrm{,}    
\end{equation}
for some $U_{i_1}^2 \in \mathcal{S}^2$. Next, we repeat the argument for $\mathcal{S}^2_{i_1}$ and get 
\begin{equation}
||(\pi_{\lambda}(e)-\pi_{\lambda}(U_{i_2}^2))w||\geq b_{i_1}(\lambda) \geq b_1(\lambda) \mathrm{,}   
\end{equation}
for some $U_{i_2}^2 \in \mathcal{S}^2_{i_1}$. We proceed in this manner until $m=k-2$, which gives
\begin{equation}
||(\pi_{\lambda}(e)-\pi_{\lambda}(U_{i_{k-1}}^2))w||\geq b_{i_1 i_2 \ldots i_{k-2}}(\lambda) \geq b_{k-2}(\lambda) \mathrm{,}    
\end{equation}
for some $U_{i_{k-1}}^2 \in \mathcal{S}^2_{i_1 \ldots i_{k-2}}$.

This way we obtain bounds for each pair generators except for one pair $\{U^2_{i_k}, U^{-2}_{i_k}\}$, where $1 \leq i_k\leq k$ is the remaining index.
Thus, using (\ref{eq:pm}), for all $i_m$ with $m \in \{1, \ldots, k-1\}$ we have
\begin{equation}
||(\pi_{\lambda}(U_{i_m})+\pi_{\lambda}( U_{i_m}^{-1})))w|| \leq \sqrt{4-b_{m-1}^2(\lambda)}    
\end{equation}
provided that
\begin{equation}
\epsilon_{m} \leq \frac{\sqrt{2 \left(1- \frac{\imath_{\lambda}}{d_{\lambda}}\right)}}{C ||\lambda||_1}, \epsilon_{m} < 1 \mathrm{.}    
\end{equation}
For $\ell_{m-1} \gg 1$, the good approximation is
\begin{equation}
\sqrt{4-b^2_{m-1}(\lambda)} \leq 2 - \left( \frac{b_{m-1}(\lambda)}{2}\right)^2 \mathrm{.}    
\end{equation}
Hence,
\begin{equation}
\label{eq:b1}
\mathrm{gap}_t(\mathcal{S}) \geq \mathrm{min}_{\lambda \in \tilde{\Lambda}_{t}} \frac{1}{8 k} \sum_{m=0}^{k-2} b_m^2(\lambda) > 0 \mathcal{.}     
\end{equation}

Using similar argument by considering only $\mathcal{S}^2_{i_1 \ldots i_{k-2}}$ we have the following, weaker bound
\begin{equation}
\label{eq:b2}
\mathrm{gap}_t(\mathcal{S}) \geq \mathrm{min}_{\lambda \in \tilde{\Lambda}_{t}} \frac{k-1}{8k}b_{k-2}^2(\lambda) > 0 \mathcal{.}
\end{equation}
Indeed, comparing (\ref{eq:b1}) and (\ref{eq:b2}) we have the inequality
\begin{equation}
\label{eq:comp2}
 \frac{1}{8k} \sum_{m=0}^{k-2} b_m^2(\lambda) \geq  \frac{k-1}{8k} b_{k-2}^2(\lambda) \mathrm{.}  \end{equation}
Moreover, the ratio between LHS and RHS of (\ref{eq:comp2}) is
\begin{equation}
 \frac{\frac{1}{k-1}  \sum_{m=0}^{k-2} b_m^2(\lambda) }{b_{k-2}^2(\lambda) } \mathrm{,}
\end{equation}
i.e. it is ratio between the average of a nonincreasing sequence $b_0^2, b_1^2, \ldots, b^2_{k-2}$ and its smallest element $ b^2_{k-2}$. Since we expect this sequence to (generically) quickly decrease, we suppose that the bound (\ref{eq:b1}) is (relatively) much better than (\ref{eq:b2}), at least generically.



It remains to somehow simultaneously bound the coefficients $b_m(\lambda)$ for all $\lambda \in \tilde{\Lambda}_t$. Since $\ell_{m} \leq \mathrm{diam}_{\epsilon}(G,\mathcal{S}^2_{m})$, from (\ref{eq:b1}), (\ref{eq:b}) we obtain the bound for the gap from the diameter

\begin{equation}
\mathrm{gap}_t(\mathcal{S}) \geq \mathrm{min}_{\lambda \in \tilde{\Lambda}_t} \frac{1}{8k} \sum_{m=0}^{k-2}  \left(\sqrt{2 \left(1- \frac{\imath_{\lambda}}{d_{\lambda}}\right)}-C ||\lambda||_1 \epsilon_{m}(\lambda)\right)^2 \frac{1}{ \mathrm{diam}_{\epsilon_{m}(\lambda)}(G,\mathcal{S}^2_{m})^2}    
\end{equation}
valid for
\begin{equation}
0< \epsilon_{m}(\lambda) \leq \frac{\sqrt{2 \left(1- \frac{\imath_{\lambda}}{d_{\lambda}}\right)}}{C ||\lambda||_1}, \epsilon_{m}(\lambda) < 1 \mathrm{,}    
\end{equation}
which can be weakened to the following simplified bound

\begin{equation}
\mathrm{gap}_t(\mathcal{S}) \geq \frac{1}{8k} \sum_{m=0}^{k-2}  \left(1-2C t \epsilon_{m}\right)^2 \frac{1}{ \mathrm{diam}_{\epsilon_{m}}(G,\mathcal{S}^2_{m})^2}    
\end{equation}
valid for
\begin{equation}
0< \epsilon_{m} \leq \frac{1}{2Ct}, \epsilon_{m} < 1 \mathrm{.}    
\end{equation}

We have a trade-off between the contribution of $\epsilon_{m}$ to the numerator of multiplicative term (the smaller the $\epsilon_{m}$ the better) and to the diameter (the larger the $\epsilon_{m}$ the better). 

Because we do not know how $\mathrm{diam}_{\epsilon_{m}}(G,\mathcal{S}_{m}^2)$ depends on $\epsilon_{m}$, in order to proceed we can use
Solovay-Kitaev theorem for $\mathcal{S}^2_m$ to bound
\begin{equation}
\mathrm{diam}_{\epsilon_{m}}(G,\mathcal{S}_{m}^2) \leq A_{m} \cdot \mathrm{log}^{c}\left(\frac{1}{c_s^2 \epsilon_{m}}\right), \quad A_{m}=\frac{1}{\left[ 2 \mathrm{log} \left(\frac{1}{c_s \epsilon_{0,m}}\right)\right]^c} \ell_{0, m}    
\end{equation}
where $c=\mathrm{log}(5)/\mathrm{log}(3/2) \approx 4$, $c_s$ is some constant ($c_s = d+2 + \mathcal{O}(\epsilon) $), $\epsilon_{0,m}$ is the $\epsilon$ of initial approximation in Solovay-Kitaev algorithm and $\ell_{0, m}$ is the word length of this approximation. Thus,

\begin{equation}
\label{eq:pre1}
\mathrm{gap}_t(\mathcal{S}) \geq \frac{1}{8k} \sum_{m=0}^{k-2}  \left(\frac{1-2C t \epsilon_{m}}{A_{m} }\right)^2 \mathrm{log}^{-2c}(c_s^{-2}\epsilon_{m}^{-1})
\end{equation}
for any
\begin{equation}
0< \epsilon_{m} \leq \frac{1}{2Ct}, \epsilon_{m} < 1 \mathrm{.}   
\end{equation}
We can bound $\ell_{0,m}$ by (\ref{eq:bmo}),

\begin{equation}
\label{eq:bmo2}
\ell_{0,m} \geq \frac{(d^2-1)(2 \mathrm{log}\left(\frac{1}{\epsilon_{0,m}}\right) + \mathrm{log}(4 C_b^{3/2}d)) + \mathrm{log}(32)}{\mathrm{gap}_{t_{0,m}}(\mathcal{S}^2_m)} \mathrm{,}    
\end{equation}
for $C_b=9 \pi$ and $t_{0,m} \geq 5 d^{5/2} /\epsilon_{0,m} \cdot \tau(\epsilon_{0,m},d)$.

For simplicity, we set the common $\epsilon_{0,m}=\epsilon_0$ and put $\epsilon_m=1/(4Ct)$, which yields

\begin{equation}
\mathrm{gap}_t(\mathcal{S}) \geq \frac{1}{32k} \sum_{m=0}^{k-2}  \frac{1}{ A_{m}^2}  \mathrm{log}^{-2c}(4 c_s^{-2} C t)=\frac{1}{32k} \sum_{m=0}^{k-2} \frac{1}{\ell_{0,m}^2} \left[2 \mathrm{log}(c_s^{-1} \epsilon_0^{-1})\right]^{2c}   \mathrm{log}^{-2c}(4 c_s^{-2} C t) \mathrm{,}
\end{equation}
and by setting the common scale $t_{0,m}=t_0\coloneqq 5d^{5/2}/{\epsilon_0} \cdot \tau(\epsilon_0, d)$ and using (\ref{eq:bmo2}) to set the value of $\ell_{0,m}$ we obtain
\begin{equation}
\mathrm{gap}_t(\mathcal{S}) \geq \frac{1}{32k} \sum_{m=0}^{k-2} \frac{\mathrm{gap}^2_{t_{0}}(\mathcal{S}^2_m)}{\left((d^2-1)(2 \mathrm{log}(\epsilon_0^{-1}) + \mathrm{log}(4 C_b^{3/2}d)) + \mathrm{log}(32)\right)^2} \left[\frac{2 \mathrm{log}(c_s^{-1} \epsilon_0^{-1})}{ \mathrm{log}(4 c_s^{-2} C t)}\right]^{2c}   \mathrm{,}
\end{equation}
which can be rewritten as

\begin{equation}
\label{eq:bf}
\mathrm{gap}_t(\mathcal{S}) \geq  \alpha \cdot g_{t_0}(\mathcal{S}) \cdot \mathrm{log}^{-2c}(\beta t)  \mathrm{,}
\end{equation}
where $\alpha$ and $\beta$ are
\begin{equation}
\label{eq:ab}
\alpha \coloneqq \frac{\left[2 \mathrm{log}(c_s^{-1} \epsilon_0^{-1})\right]^{2c}}{16 \cdot \left((d^2-1)(2 \mathrm{log}(\epsilon_0^{-1}) + \mathrm{log}(4 C_b^{3/2}d)) + \mathrm{log}(32)\right)^2} \mathrm{,}  \quad \beta \coloneqq \frac{4C}{c_s^2}   \mathrm{,}
\end{equation}
and
\begin{equation}
\label{eq:gt0}
g_{t_0}(\mathcal{S}) \coloneqq  \frac{1}{|\mathcal{S}|}\sum_{m=0}^{|\mathcal{S}|/2-2} \mathrm{gap}^2_{t_{0}}(\mathcal{S}^2_m) \mathcal{.}
\end{equation}

Finally, we can redefine $\mathrm{gap}_{t_{0}}(\mathcal{S}^2_m)$ to be the smallest value of a gap at scale $t_0$ over all sets $\mathcal{S}^2_{i_1 \ldots i_{m}}$,

\begin{equation}
\label{eq:gapt0r}
\mathrm{gap}_{t_{0}}(\mathcal{S}^2_m) \coloneqq \mathrm{argmin}_{i_1, \ldots, i_m} \mathrm{gap}_{t_0}(\mathcal{S}^2_{i_1 \ldots i_{m}})   
\end{equation}
and this way $g_{t_0}(\mathcal{S})$ can be determined numerically by the calculations at scale
\begin{equation}
\label{eq:t0}
t_0\coloneqq 5d^{5/2}/{\epsilon_0} \cdot \tau(\epsilon_0, d) \mathrm{.}
\end{equation}









\section*{Acknowledgments}
This research was funded by the National Science Centre, Poland under the grant OPUS: UMO-2020/37/B/ST2/02478. 

\printbibliography

@article{varju13,
    author = "P. P. Varj\'u",
    title = "Random walks in compact groups",
    year = "2013",
    journal = "Documenta Mathematica",
    volume= "18",
    pages = "1137--1175"
}

@article{oszmaniec22,
    author = "M. Oszmaniec and A. Sawicki and M. Horodecki",
    title = "Epsilon-nets, unitary designs and random quantum circuits",
    year = "2022",
    journal = "IEEE Transactions on Information Theory",
    volume= "68",
    issue = "2",
    pages =  "989 -- 1015"
}

@article{kesten59,
    author = "H. Kesten",
    title = "Symmetric Random Walks on Groups",
    year = "1959",
    journal = "Transactions of the American Mathematical Society",
    volume = "92",
    issue = "2",
    pages = "336--354"
}

@article{gamburd08,
    author = "J. Bourgain and A. Gamburd",
    title = "On the spectral gap for finitely-generated subgroups of SU(2)",
    year = "2008",
    journal = "Inventiones Mathematicae",
    volume = "171",
    pages = "83–-121"
}

@article{gamburd12,
    author = "J. Bourgain and A. Gamburd",
    title = "A spectral gap theorem in SU(d)",
    year = "2012",
    journal = "Journal of the European Mathematical Society",
    volume = "14",
    issue = "5",
    pages = "1455--1511"
}

@article{lps86,
    author = "A. Lubotzky and R. Phillips and P. Sarnak",
    title = "Hecke operators and distributing points on the Sphere I",
    year = "1986",
    journal = "Communications on Pure and Applied Mathematics. Supplement: Proceedings of the Symposium on Frontiers of the Mathematical Sciences: 1985.",
    volume = "39",
    issue = "S1",
    pages = "S149-S186"
}

@article{lps87,
    author = "A. Lubotzky and R. Phillips and P. Sarnak",
    title = "Hecke operators and distributing points on S2. II",
    year = "1987",
    journal = "Communications on Pure and Applied Mathematics",
    volume = "40",
    issue = "4",
    pages = "401-420"
}

@article{harrow02,
    author = "A. W. Harrow and B. Recht and I. L. Chuang",
    title = "Efficient Discrete Approximations of Quantum Gates",
    year = "2002",
    journal = "Journal of Mathematical Physics",
    volume = "43",
    pages = "4445-4451"
}

@article{sawicki17,
    author = "A. Sawicki and K. Karnas",
    title = "Criteria for universality of quantum gates",
    year = "2017",
    journal = "Physical Review A",
    volume = "95",
    pages = "062303"
}

@article{sawicki_s17,
    author = "A. Sawicki and K. Karnas",
    title = "Universality of Single-Qudit Gates",
    year = "2017",
    journal = "Annales Henri Poincaré",
    volume = "18",
    pages = "3515–3552"
}

@article{sawicki21,
    author = "A. Sawicki and L. Mattioli and Z. Zimborás",
    title = "How to check universality of quantum gates?",
    year = "2021",
    journal = "arXiv:2111.03862"
}

@article{preskill2018,
    author = "J. Preskill",
    title = "Quantum Computing in the NISQ era and beyond",
    year = "2018",
    journal = "Quantum",
    volume = "2",
    issue = "79"
    
}

@article{boxio2018,
    author = "S. Boixo and S. V. Isakov and V. N. Smelyanskiy and R. Babbush and N. Ding and Z. Jiang and M. J. Bremner and J. M. Martinis and H. Neven",
    title = "Characterizing Quantum Supremacy in Near-Term Devices",
    year = "2018",
    journal = "Nature Physics",
    volume = "14",
    pages = "595--600"
}

@article{montanaro2017,
    author = "A. W. Harrow and A. Montanaro",
    title = "Quantum Computational Supremacy",
    year = "2017",
    journal = "Nature",
    volume = "549",
    pages = "203--209"
}

@article{bocharov2013,
    author = "A. Bocharov and Y. Gurevich and K. M. Svore",
    title = "Efficient decomposition of single-qubit gates into V basis circuits",
    year = "2013",
    journal = "Physical Review A",
    volume = "88",
    issue = "1"
}

@article{kliuch2016,
    author = "V. Kliuchnikov and D. Maslov and M. Mosca",
    title = "Practical Approximation of Single-Qubit Unitaries by Single-Qubit Quantum Clifford and T Circuits",
    year = "2016",
    journal = "IEEE Transactions on Computers",
    volume = "65",
    pages = "161--172"
}

@article{selinger2015,
    author = "P. Selinger",
    title = "Efficient Clifford+T approximation of single-qubit operators",
    year = "2015",
    journal = "Quantum Information and Computation",
    volume = "15",
    issue = "1-2",
    pages = "159--180"
}

@article{sarnak2015,
    author = "P. Sarnak",
    title = "Letter to Scott Aaronson and Andy Pollington on the Solovay-Kitaev theorem",
    year = "2015"
}

@article{aharonov2008,
    author = "D. Aharonov and M. Ben-Or",
    title = "Fault-Tolerant Quantum Computation with Constant Error Rate",
    year = "2008",
    journal = "SIAM Journal on Computing",
    volume = "38",
    issue = "4",
    pages = "1207--1282"
}

@article{knill1998,
    author = "E. Knill and R. Laflamme and W. H. Zurek",
    title = "Resilient Quantum Computation",
    year = "1998",
    journal = "Science",
    volume = "279",
    issue = "5349",
    pages ="342–-345"
}

@article{kitaev2003,
    author = "A. Y. Kitaev",
    title = "Fault-tolerant quantum computation by anyons",
    year = "2003",
    journal = "Science",
    volume = "303",
    issue = "1",
    pages ="2--30"
}

@article{szarek1998,
    author = "S. J. Szarek",
    title = "Metric Entropy of Homogeneous Spaces",
    year = "1998",
    journal = "Banach Center Publications",
    volume="43"
}

@book{kitaev2002,
    author = "A. Y. Kitaev and A. Shen and M. N. Vyalyi",
    title = "Classical and quantum computation",
    year = "2002",
    publisher = "Providence, Rhode Island: American Mathematical Society",
    isbn="0-8218-2161-X"

}

\end{document}